\documentclass[a4paper,11pt]{article}

\usepackage[toc,page]{appendix}

\usepackage{amsthm}

\usepackage{amsmath}
\usepackage{latexsym}
\usepackage{amssymb}
\usepackage{verbatim}
\usepackage{setspace}
\usepackage{qtree}
\usepackage{enumitem}
\usepackage{xcolor}
\usepackage{url}
\usepackage{stmaryrd}

\usepackage{authblk}

\usepackage[margin=1.4in]{geometry}

\usepackage{tikz}
\usetikzlibrary{arrows,decorations.pathmorphing,backgrounds,positioning,fit}

\author{Reijo Jaakkola}
\affil{University of Helsinki and Tampere University, Finland}

\date{}

\begin{document}

\setlength\abovedisplayskip{3pt}
\setlength\belowdisplayskip{3pt}
\title{Ordered fragments of first-order logic}

\theoremstyle{plain}
\newtheorem{theorem}{Theorem}[section]
\newtheorem{lemma}[theorem]{Lemma}
\newtheorem{corollary}[theorem]{Corollary}
\newtheorem{proposition}[theorem]{Proposition}
\theoremstyle{definition}
\newtheorem{definition}[theorem]{Definition}
\newtheorem{remark}[theorem]{Remark}
\newtheorem{example}[theorem]{Example}

\maketitle

\begin{abstract}
\noindent
   Using a recently introduced algebraic framework for classifying fragments of first-order logic, we study the complexity of the satisfiability problem for several ordered fragments of first-order logic, which are obtained from the ordered logic and the fluted logic by modifying some of their syntactical restrictions.
\end{abstract}

%%%%%%%%%%%%%%
%%%%%%%%%%%%%%
%%%%%%%%%%%%%%
%%%%%%%%%%%%%%

\section{Introduction}

Informally speaking, a fragment of first-order logic is called ordered, if the syntax of the fragment restricts permutations of variables (with respect to some ordering of the variables) and the order in which the variables are to be quantified. To give an example of a fragment that belongs to this family of logics, we will mention the so-called fluted logic which has received some attention quite recently \cite{flutedlidiatendera, flutedequality, Voigt2019SeparatenessOV}. Roughly speaking, in fluted logic the order in which the variables are being quantified should be the same as the order in which these variables occur in atomic formulas. To illustrate this restriction, consider the following sentence.
\[\forall v_1 (P(v_1) \to \exists v_2 (R(v_1,v_2) \land \forall v_3 S(v_1,v_2,v_3))).\]
This sentence belongs to the fluted logic, since the variables $v_1,v_2$ and $v_3$ are quantified away in the correct order.

One motivation for the study of ordered fragments of first-order logic comes from the fact that they are orthogonal in expressive power with respect to other well-known fragments of first-order logic, such as the guarded fragments. For instance, the formula $\forall v_1 \exists v_2 \forall v_3 R(v_1,v_2,v_3)$ is clearly ordered, but it expresses a property that is, for example, neither expressible in any of the guarded fragments nor in the two-variable fragment. Thus one should expect that the study of these fragments forces us to come up with novel techniques.

Since the syntax of ordered logics restricts heavily the permutations of variables and the order in which the variables are quantified, their syntax can often be presented naturally in a variable-free way. For instance, if we agree that variables must be quantified away with respect to their indexes, then the expression $\forall \exists R$ corresponds unambiguously to the formula $\forall v_1 \exists v_2 R(v_1,v_2)$. Indeed, the fluted logic was originally discovered by Quine as a by-product of his attempts to present the full syntax of first-order logic using a variable-free syntax \cite{quine60, quine72}. Interestingly, this approach was also adopted in the recent paper \cite{flutedequality}, where the fluted logic was presented using its variable-free syntax.

Recently a research program was introduced in \cite{preprintofthis, preprintofthis2, games19} for classifying fragments of first-order logic within an algebraic framework. In a nutshell, the basic idea is to identify fragments of first-order logic with finite algebraic signatures (for more details, see the next section). The algebraic framework naturally suggests the idea of defining logics with limited permutations, and hence it is well suited for representing various ordered fragments of first-order logic. The purpose of this article is to apply this framework to study how the complexities of the ordered fragments change, if we make small modifications in their syntax.

The first logic that we consider in this article is the so-called ordered logic, which is a fragment of the fluted logic (for a formal definition see section 3). The satisfiability problem for this fragment was proved to be in \textsc{Pspace} in \cite{herzigordered}, which should be compared with the fact that the satisfiability problem for fluted logic is \textsc{Tower}-complete, which was proved in \cite{flutedlidiatendera}. We note that the problem is in fact \textsc{Pspace}-complete, and the proof for the \textsc{Pspace}-hardness will be given in section 5.

Now it is natural to ask whether one could extend the syntax of ordered logic while maintaining the requirement that the complexity of the satisfiability problem remains relatively low, and this is the first question that we attempt to answer using the aforementioned algebraic framework. We will formalize different minimal extension of the ordered logic using additional algebraic operators and study the complexities of the resulting logics. The picture that emerges from our results seems to suggest that even if one modifies the syntax of the ordered logic in a very minimal way, the resulting logics will most likely have much higher complexity. For instance, if we relax even slightly the order in which the variables can be quantified, the resulting logic will have \textsc{NexpTime}-complete satisfiability problem. However, there are also exceptions to this rule, since the complexity of ordered logic with equality turns out to be the same as the complexity of the regular ordered logic.

Motivated by the recent study of one-dimensional guarded fragments conducted in \cite{Kieronski2019OnedimensionalGF}, we will also study the one-dimensional fragments of fluted logic and the ordered logic. Intuitively a logic is called one-dimensional if quantification is limited to applications of blocks of existential (universal) quantifiers such that at most one variable remains free in the quantified formula. Imposing the restriction of one-dimensionality to fluted logic and ordered logic decreases quite considerably the complexity of the underlying logics: the complexity of the one-dimensional fluted logic is \textsc{NexpTime}-complete while the complexity of the one-dimensional ordered logic with equality is \textsc{NP}-complete. In the case of the fluted logic we are able to add some further algebraic operators into its syntax without increasing its complexity.

We will also prove that several natural extensions of the ordered logic and the fluted logic are undecidable. First, for the ordered logic we are able to show that if we allow variables to be quantified in an arbitrary order, then the resulting logic is undecidable. Secondly, we are able to show that if we remove restrictions on how the variables in the atomic formulas can be permuted in the one-dimensional fluted logic, then the resulting logic is undecidable. Finally, in the case of the full fluted logic, we can show that if we relax only slightly the way variables can be permuted and the order in which variables can be quantified, then the resulting logic is undecidable.

\section{Algebraic way of presenting logics}

The purpose of this section is to present the algebraic framework introduced in \cite{preprintofthis, preprintofthis2, games19} for classifying fragments of FO. We will be working with purely relational vocabularies with no constants and function symbols. In addition we will not consider vocabularies with $0$-ary relational symbols. Throughout this article we will use the convention where the domain of a model $\mathfrak{A}$ will be denoted by the set $A$.

Let $A$ be an arbitrary set. As usual, a $k$-\textbf{tuple} over $A$ is an element of $A^k$. Given a non-negative integer $k$, a $k$-ary \textbf{AD-relation} over $A$ is a pair $T = (X,k)$, where $X \subseteq A^k$. Here 'AD' stands for arity-definite. Given a $k$-ary AD-relation $T = (X,k)$ over $A$, we will use $(a_1,...,a_k) \in T$ to denote $(a_1,...,a_k) \in X$. Given an AD-relation $T$, we will use $ar(T)$ to denote its arity. We call $(\varnothing,k)$ the \textbf{empty} $k$-ary AD-relation. We will also write $\bot_k^A := (\varnothing,k)$ to emphasize that the empty relation is over the set $A$. Similarly, we will write $\top_k^A := (A^k,k)$.

Given a set $A$, we will use $\mathrm{AD}(A)$ to denote the set of all AD-relations over $A$. If $T_1,...,T_k \in \mathrm{AD}(A)$, then the tuple $(A,T_1,...,T_k)$ will be called an \textbf{AD-structure} over $A$. A bijection $g:A \to B$ is an \textbf{isomorphism} between AD-structures $(A,T_1,...,T_k)$ and $(B,S_1,...,S_k)$, if for every $1\leq \ell \leq k$ we have that $ar(T_\ell) = ar(S_\ell)$, and $g$ is an ordinary isomorphism between the relational structures $(A,rel(T_1),...,rel(T_k))$ and $(B,rel(S_1),...,rel(S_k))$, where $rel(T)$ denotes the underlying relation of an AD-relation. 

The following definition was introduced in \cite{preprintofthis2}, where it was called arity-regular relation operator.

\begin{definition}
     A $k$-ary \textbf{relation operator} F is a mapping which associates to each set $A$ a function $F^A:\mathrm{AD}(A)^k \to \mathrm{AD}(A)$ and which satisfies the following requirements.
     \begin{enumerate}
         \item The operator $F$ is isomorphism invariant in the sense that if the AD-structures $(A,T_1,...,T_k)$ and $(B,S_1,...,S_k)$ are isomorphic via $g$, then the AD-structures $(A,F^A(T_1,...,T_k))$ and $(B,F^B(S_1,...,S_k))$ are likewise isomorphic via $g$.
         \item There exists a function $\sharp : \mathbb{N}^k \to \mathbb{N}$ so that for every AD-structure $(A,T_1,...,T_k)$ we have that the arity of the AD-relation $F^A(T_1,...,T_k)$ is $\sharp (ar(T_1),...,ar(T_k))$. In other words the arity of the output AD-relation is always determined fully by the sequence of arities of the input AD-relations.
     \end{enumerate}
\end{definition}

Given a set of relation operators $\mathcal{F}$ and a vocabulary $\tau$, we can define a language $\mathrm{GRA}(\mathcal{F})[\tau]$ as follows, where $R\in \tau$ and $F \in \mathcal{F}$, and $\bot$ and $\top$ are $0$-ary relation operators which are defined in the obvious way:
\[\mathcal{T} ::= \bot \mid \top \mid R \mid \mathrm{F}\underbrace{(\mathcal{T},...,\mathcal{T})}_{ar(\mathrm{F}) \text{ times}}.\]
Here 'GRA' stands for general relational algebra. We sometimes use the infix notation instead of the prefix notation, if the infix notation is more conventional. For example, if $\cap$ denotes a relation operator, then instead of writing $\cap (\mathcal{T},\mathcal{P})$, we will write $(\mathcal{T} \cap \mathcal{P})$. Furthermore we will drop the brackets in the case where F is unary operator.

If the underlying vocabulary $\tau$ is clear from context or irrelevant, we will write $\mathrm{GRA}(\mathcal{F})$ instead of $\mathrm{GRA}(\mathcal{F})[\tau]$. The members of $\mathrm{GRA}(\mathcal{F})$ will be referred to as \textbf{terms}. In the case where $\mathcal{F}$ is a finite set $\{F_1,..,F_n\}$, we will use $\mathrm{GRA}(F_1,...,F_n)$ to denote $\mathrm{GRA}(\{F_1,...,F_n\})$.

Given a model $\mathfrak{A}$ of vocabulary $\tau$ and $\mathcal{T} \in \mathrm{GRA}(\mathcal{F})[\tau]$, we define its \textbf{interpretation} $\llbracket \mathcal{T} \rrbracket_\mathfrak{A}$ recursively as follows:
\begin{enumerate}
    \item If $\mathcal{T} = \bot$, then we define $\llbracket \mathcal{T} \rrbracket_\mathfrak{A} = \bot^A := \bot^A_0$. Similarly, if $\mathcal{T} = \top$, then we define $\llbracket \mathcal{T} \rrbracket_\mathfrak{A} = \top^A := \top^A_0$.
    \item If $\mathcal{T} = R \in \tau$, then we define $\llbracket R \rrbracket_\mathfrak{A} = (R^\mathfrak{A},ar(R))$.
    \item If $\mathcal{T} = \mathrm{F}(\mathcal{T}_1,...,\mathcal{T}_k)$, then we define 
    \[\llbracket \mathcal{T} \rrbracket_\mathfrak{A} = \mathrm{F}_A(\llbracket \mathcal{T}_1 \rrbracket_\mathfrak{A},...,\llbracket \mathcal{T}_k \rrbracket_\mathfrak{A})\]
\end{enumerate}
Note that the interpretation of a term over $\mathfrak{A}$ is an AD-relation over $A$. The arity of this AD-relation is called the arity of the term $\mathcal{T}$ and we will denote it by $ar(\mathcal{T})$. Note that by definition the arity of the output relation is independent of the underlying model, which guarantees that $ar(\mathcal{T})$ is well-defined. 

Given two $k$-ary terms $\mathcal{T}$ and $\mathcal{P}$ over the same vocabulary, we say that $\mathcal{T}$ is \textbf{contained} in $\mathcal{P}$, if for every model $\mathfrak{A}$ over $\tau$ and for every $(a_1,...,a_k) \in A^k$ we have that if $(a_1,...,a_k) \in \llbracket \mathcal{T} \rrbracket_\mathfrak{A}$ then $(a_1,...,a_k) \in \llbracket \mathcal{P} \rrbracket_\mathfrak{A}$. We will denote this by $\mathcal{T} \models \mathcal{P}$. If $\mathcal{T}$ is a $0$-ary term and $\mathfrak{A}$ is a model so that $\llbracket \mathcal{T} \rrbracket_\mathfrak{A} = \top_0^A$, then we denote this by $\mathfrak{A} \models \mathcal{T}$. Given a $k$-ary term $\mathcal{T}$, we say that $\mathcal{T}$ is \textbf{satisfiable} if there exists a model $\mathfrak{A}$ so that $\llbracket \mathcal{T} \rrbracket_\mathfrak{A}$ is not the empty relation.

We will conclude this section by showing how we can compare the expressive power of algebras with fragments of $\mathrm{FO}$. Let $k\geq 0$ and consider an FO-formula $\varphi(v_{i_1},...,v_{i_k})$, where $(v_{i_1},...,v_{i_k})$ lists all the free variables of $\varphi$, and $i_1 < ... < i_k$. If $\mathfrak{A}$ is a suitable model, then $\varphi$ defines the following AD-relation over $A$
\[\llbracket \varphi \rrbracket_\mathfrak{A} = (\{(a_1,...,a_k) \mid \mathfrak{A} \models \varphi(a_1,...,a_k)\},k)\]
Given a $k$-ary term $\mathcal{T}$ and FO-formula $\varphi(v_{i_1},...,v_{i_k})$ over the same vocabulary, we say that $\mathcal{T}$ is \textbf{equivalent} with $\varphi$ if for every model $\mathfrak{A}$ we have that $\llbracket \mathcal{T} \rrbracket_\mathfrak{A} = \llbracket \varphi \rrbracket_\mathfrak{A}$.

\begin{definition}
    Let $\mathcal{F}$ be a set of relation operators and let $\mathcal{L} \subseteq \mathrm{FO}$. \begin{enumerate}
        \item We say that $\mathrm{GRA}(\mathcal{F})$ and $\mathcal{L}$ are \textbf{equivalent}, if for every $\mathcal{T} \in \mathrm{GRA}(\mathcal{F})$ there exists an equivalent formula $\varphi \in \mathcal{L}$, and conversely for every formula $\varphi \in \mathcal{L}$ there exists an equivalent term $\mathcal{T} \in \mathrm{GRA}(\mathcal{F})$.
        \item We say that $\mathrm{GRA}(\mathcal{F})$ and $\mathcal{L}$ are \textbf{sententially equivalent}, if for every $0$-ary term $\mathcal{T} \in \mathrm{GRA}(\mathcal{F})$ there exists an equivalent sentence $\varphi \in \mathcal{L}$, and conversely for every sentence $\varphi \in \mathcal{L}$ there exists an equivalent $0$-ary term $\mathcal{T} \in \mathrm{GRA}(\mathcal{F})$.
    \end{enumerate}
\end{definition}

\section{Relevant fragments and complexity results}

The purpose of this section is to define the relevant FO-fragments that we are going to study and to present the main complexity results that we are able to obtain. Through out this section $(X,k)$ and $(Y,\ell)$ are AD-relations over some set $A$.

We are going to start by defining formally the ordered logic $\mathrm{OL}$, which will form the backbone for the rest of fragments studied in this article.

\begin{definition}
    Let $\overline{v}_\omega = (v_1,v_2,...)$ and let $\tau$ be a vocabulary. For every $k \in \mathbb{N}$ we define sets $\mathrm{OL}^k[\tau]$ as follows.
    \begin{enumerate}
        \item Let $R \in \tau$ be an $\ell$-ary relational symbol and consider the prefix
        \[(v_1,...,v_\ell)\]
        of $\overline{v}_\omega$ containing precisely $\ell$-variables. If $k\geq \ell$, then $R(v_1,...,v_\ell) \in \mathrm{OL}^k[\tau]$.
        \item Let $\ell \leq \ell' \leq k$ and suppose that $\varphi \in \mathrm{OL}^\ell[\tau]$ and $\psi \in \mathrm{OL}^{\ell'}[\tau]$. Then $\neg \varphi, (\varphi \land \psi) \in \mathrm{OL}^k[\tau]$.
        \item If $\varphi \in \mathrm{OL}^{k+1}[\tau]$, then $\exists v_{k+1} \varphi \in \mathrm{OL}^k[\tau]$.
    \end{enumerate}
    Finally we define $\mathrm{OL}[\tau] := \bigcup_k \mathrm{OL}^k[\tau]$.
\end{definition}

\begin{remark}
    The way we have presented the syntax of $\mathrm{OL}$ here is slightly different from the way it is often presented in the literature. We chose to present the syntax this way to highlight the connection between $\mathrm{OL}$ and the fluted logic, since it is also natural to define the syntax for the latter logic inductively with respect to some parameter $k$.
\end{remark}

The syntax of this logic is somewhat involved, but it can be given a very nice algebraic characterization using just three relation operators $\{\neg,\cap,\exists\}$, which we are going to define next. Recalling that if $F$ is a relation operator, then $F^A$ denotes the function to which $F$ maps the set $A$, we can define the relation operators as follows. 

\begin{enumerate}
    \item [$\neg)$] We define $\neg^A(X,k) = (A^k\backslash X,k)$. We call $\neg$ the \textbf{complementation} operator.
    \item [$\cap)$] If $k\neq \ell$, then we define $\cap^A((X,k),(Y,\ell)) = \bot_0^A$. Otherwise we define 
    \[\cap^A((X,k),(Y,\ell)) = (X \cap Y, k).\]
    We call $\cap$ the \textbf{intersection} operator.
    \item [$\exists)$] If $k = 0$, then we define $\exists^A(X,k) = (X,k)$. Otherwise we define
    \[\exists^A(X,k) = (\{\overline{a} \mid \text{$\overline{a}b\in X$, for some $b\in A$}\},k-1).\]
    We call $\exists$ the \textbf{projection} operator.
\end{enumerate}

The following proposition establishes the promised characterization result.

\begin{proposition}
    $\mathrm{OL}$ and $\mathrm{GRA}(\neg,\cap,\exists)$ are sententially equiexpressive.
\end{proposition}
\begin{proof}
    We will focus on translating $\mathrm{OL}$ to $\mathrm{GRA}(\neg,\cap,\exists)$. As a first step, we will show how to translate every sentence in $\varphi\in \mathrm{OL}$ into an equivalent one which satisfies the following property. If $(\psi \circ \chi)$, where $\circ \in \{\lor,\land\}$, is a subformula of $\varphi$, then $\mathrm{Free}(\psi) = \mathrm{Free}(\chi)$, where $\mathrm{Free}(\psi)$ denotes the set of free variables of $\psi$. This can be achieved by pushing quantifiers inwards as follows. Suppose that we have a subformula of the form $Qv_k(\psi \circ \chi)$, where $v_k \in \mathrm{Free}(\psi) \cup \mathrm{Free}(\psi)$. Now we know that if $v_k \in \mathrm{Free}(\psi)$ and $v_k \in \mathrm{Free}(\chi)$, then $\mathrm{Free}(\psi) = \mathrm{Free}(\chi)$, since $v_k$ must be the variable among the free variables of $\psi$ and $\chi$ with the largest index. On the other hand if we have that $v_k$ occurs as a free variable only in one the formulas, say $\psi$, then $Qv_k(\psi \circ \chi)$ is equivalent to $(Qv_k \psi \circ \chi)$. Continuing this way it is clear that we achieve an equivalent sentence with the desired property.
    
    We then translate sentences of ordered logic to algebraic terms. Suppose that $\varphi \in \mathrm{OL}$ is a sentence which satisfies the above property. Formulas which have the form $R(v_1,...,v_\ell)$ are translated to $R$. Suppose then that we have translated $\psi$ to $\mathcal{T}$ and $\chi$ to $\mathcal{S}$. Then we can translate $(\psi \land \chi)$ to $(\mathcal{T} \cap \mathcal{S})$, $\neg \psi$ to $\neg \mathcal{T}$ and $\exists v_k \psi$ to $\exists \mathcal{T}$. In the first case we used the fact that $\mathrm{Free}(\psi) = \mathrm{Free}(\chi)$ and in the last case we used the fact that $v_k$ must be the free variable of $\psi$ with the largest index.
\end{proof}

\begin{remark}
    This characterization highlights another interesting restriction that OL imposes implicitly to its sentences: we can form boolean combinations of formulas only if they share the same set of free variables. In the literature this restriction is referred to as the \emph{uniformity} requirement and other examples of logics which satisfy the uniformity requirement are the uniform-one dimensional logic \cite{U1a} and the binding fragments of first-order logic \cite{Mogavero}.
\end{remark}

The complexity of $\mathrm{OL}$ is rather low and thus it is natural to ask how it changes if we add additional operators to the syntax of the logic. The first operator that is studied in this article is the operator $E$, which we define as follows.

\begin{enumerate}
    \item [$E)$] If $k < 2$, then we define $E^A(X,k) = (X,k)$. Otherwise we define
    \[E^A(X,k) = (\{(a \in X \mid a_{k-1} = a_k\},k).\]
    We call $E$ the \textbf{equality} operator.
\end{enumerate}

It turns out that the addition of equality does not increase the complexity of ordered logic. In our proof for the \textsc{Pspace} upper bound, it will be convenient to extend the ordered logic with an additional operator $I$, which we define as follows. 

\begin{enumerate}
    \item [$I)$] If $k \leq 1$, then we define $I^A(X,k) = (X,k)$, and otherwise we define
    \[I^A(X,k) = (\{\overline{a} \in A^{k-1} \mid \overline{a}a_k \in X\}, k - 1).\]
    We call $r$ the substitution operator.
\end{enumerate}
Note that if $ar(\mathcal{T}) \geq 2$, then $I\mathcal{T}$ is equivalent with $\exists E \mathcal{T}$, and hence the operator $I$ is definable in the algebra $\mathrm{GRA}(E,\neg,\cap,\exists)$.

Although the equality operator does not increase the complexity of OL, we will prove that lifting the two syntactical restrictions of OL in a minimal way will result in an increase in the complexity. More concretely, adding either of the following two operators to OL will result in a \textsc{NexpTime}-hard logic.

\begin{enumerate}
    \item [$s)$] If $k < 2$, then we define $s^A(X,k) = (X,k)$. Otherwise we define
    \[s^A(X,k) = (\{(a_1,...,a_{k-2},a_k,a_{k-1}) \mid (a_1,...,a_k) \in X\},k).\]
    We call $s$ the \textbf{swap} operator.
    \item [$C)$] If $k\neq 1$ and $\ell \leq 1$, then we define $C^A((X,k),(Y,\ell)) = \bot_A^0$. In the case where $1 = k \leq \ell$ (the case $1 = \ell \leq k$ is defined similarly) we will define
    \[C^A((Y,\ell),(X,k)) = (\{\overline{a} \in Y \mid a_\ell \in X\},\ell).\]
    We call $C$ the \textbf{one-dimensional intersection}.
\end{enumerate}

The intuition behind the swap operator is clear: it lifts in a minimal way the ordering restriction on the syntax of ordered logic. For instance, we can now express the formula $\exists v_1 R(v_1,v_2)$ using the term $\exists s R$. The one-dimensional intersection may appear to be somewhat unnatural, but the underlying intuition is that we want to lift the uniformity requirement in a minimal way. To give a concrete example of the use of one-dimensional intersection, we note that it allows us to expresses formulas such as $R(v_1,v_2) \land P(v_2)$ (this particular formula can be expressed using the term $C(R,P)$).

The other ordered fragment investigated in this article is the fluted logic FL. We will not give a formal definition for this fragment here, but instead we will introduce its algebraic characterization using the operators $\{\neg, \Dot{\cap}, \exists\}$, where $\Dot{\cap}$ is defined as follows.

\begin{enumerate}
    \item [$\Dot{\cap}$)] If $m:= \mathit{max}\{k,\ell\}$, then we define
    \begin{align*}
    \Dot{\cap}^A((X,k),(Y,\ell))
    &= \bigl(\{(a_1,\dots ,a_m)
    \mid (a_{m-k+1},\dots ,a_m) \in X\\
    &\text{ }\hspace{1.7cm}\text{ and } (a_{m-\ell + 1},\dots ,a_{m}) \in
    Y\},\, m \bigr),
    \end{align*}
    We call $\Dot\cap$ the $\textbf{suffix intersection}$.
\end{enumerate}

Intuitively, the tuples overlap on some
suffix of $(a_1,\dots , a_m)$; note here 
that when $k$ or $\ell$ is zero,
then $(a_{m+1},a_m)$ denotes the empty tuple.
Now for example the formula $R(v_1,v_2) \land P(v_2)$ is
equivalent to $R\, \Dot{\cap}\, P$
and the formula $R(v_1,v_2) \land P(v_1)$ to $s(sR\ \Dot{\cap}\, P)$. The following result was proved in \cite{preprintofthis2}.

\begin{proposition}
    $\mathrm{FL}$ and $\mathrm{GRA}(\neg,\Dot{\cap},\exists)$ are equiexpressive.
\end{proposition}

It was proved in \cite{flutedlidiatendera} that the satisfiability problem for $\mathrm{FL}$ is \textsc{Tower}-complete. The natural follow-up question is then to study what fragments of $\mathrm{FL}$ have more "feasible" complexity. In this article we approach this question by studying the so-called \emph{one-dimensional} fragment of fluted logic. To give this logic an algebraic characterization, we will need to introduce two additional operators, $\exists_1$ and $\exists_0$, which we define as follows.

\begin{enumerate}
    \item [$\exists_1$)] If $k < 2$, then we define $\exists_1^A(X,k) = (X,k)$. Otherwise we define
    \[\exists_1^A(X,k) = (\{a\in A \mid \text{There exists $\overline{b} \in A^{k-1}$ such that $a\overline{b} \in X$}\},1)\]
    \item [$\exists_0$)] If $k = 0$, then we define $\exists_0^A(X,k) = (X,k)$. Otherwise we define
    $\exists_0^A(X,k)$ to be $\top_A^0$, if $X$ is non-empty, and $\bot_A^0$, if $X$ is empty.
\end{enumerate}

We call collectively the operators $\exists_1$ and $\exists_0$ \textbf{one-dimensional projection} operators. These operators correspond to quantification which leaves at most one free-variable free. For example, if $R$ is a ternary relational symbol, then $\exists_1R$ is equivalent to $\exists v_2 \exists v_3 R(v_1,v_2,v_3)$ while $\exists_0$ is equivalent to $\exists v_1 \exists v_2 \exists v_3 R(v_1,v_2,v_3)$. Now we define the algebra $\mathrm{GRA}(\neg,\Dot{\cap},\exists_1,\exists_0)$ to be the one-dimensional fluted logic.

\begin{table}[t]
\centering
\begin{tabular}{| c | c |}
\hline
 $E,\neg,\cap,\exists_1,\exists_0$ & \textsc{NP} \\
 $E,\neg,\cap,\exists$ & \textsc{Pspace} \\ 
 $s,\neg,C,\cap,\exists$ & \textsc{NexpTime} \\ 
 $E,\neg,C,\cap,\exists$ & \textsc{NexpTime} \\ 
 $s,E,\neg,C,\cap,\exists$ & \textsc{NexpTime}-hard \\
 $s,E,\neg,\Dot{\cap}, \exists_1, \exists_0$ & \textsc{NexpTime} \\
 $p,\neg,\cap,\exists$ & $\Pi_1^0$ \\
 $p,\neg,\Dot{\cap}, \exists_1, \exists_0$ & $\Pi_1^0$ \\
 $s,\neg,\Dot{\cap},\exists$ & $\Pi_1^0$ \\
\hline
\end{tabular}
\caption{\label{tab:table-name} Complexities of the fragments.}
\end{table}

As one might expect, imposing the one-dimensionality requirement to formulas of FL will result in a logic with much lower complexity. The exact complexity of one-dimensional $\mathrm{FL}$ turns out to be \textsc{NexpTime}-complete, even for its extension with the swap and equality operators $\mathrm{GRA}(s,E,\neg,\Dot{\cap},\exists_1,\exists_0)$. In this article we also study the one-dimensional fragment of ordered logic with equality operator $\mathrm{GRA}(E,\neg,\cap,\exists)$, for which the satisfiability problem turns out to be just \textsc{NP}-complete.

Besides just decidability results, we will also prove several undecidability results. To state some of these results, we will first define the following operator $p$.

\begin{enumerate}
    \item [$p$)] If $k < 2$, then we define $p^A(X,k) = (X,k)$. Otherwise we define
    \[p^A(X,k) = (\{(a_1,...,a_k) \mid (a_k,a_1,...,a_{k-1}) \in X\},k).\]
    We call $p$ the \textbf{cyclic permutation} operator.
\end{enumerate}

It turns out that - perhaps unsurprisingly - this operator can greatly increase the expressive power of the underlying logic. For instance, together with $\exists$, this operator removes any restrictions on the order in which variables need to be quantified. Our first two undecidability results show that adding the operator $p$ to either ordered logic or one-dimensional fluted logic will lead to an undecidable logic.

Our third undecidability result is that the extension of fluted logic with swap is undecidable. Since the restriction of $\mathrm{GRA}(e,s,\neg,\Dot{\cap},\exists)$ to vocabularies of arity at most two is sententially equivalent with two-variable logic \cite{preprintofthis2}, this result seems to indicate that there does not exists a natural decidable logic that extends both the two-variable logic and the fluted logic.

Let us conclude this section by mentioning briefly two complexity results that follow immediately from the literature and which complement the picture emerging from the results listed in Table 1. First, it is easy to translate the algebra $\mathrm{GRA}(p,s,E,\neg,C,\cap,\exists_1,\exists_0)$ into the one-dimensional uniform fragment $\mathrm{UF}_1$, which was proved to be \textsc{NexpTime}-complete in \cite{U1b}. It is also not hard to see that this algebra is \textsc{NexpTime}-hard, and hence the satisfiability problem for this algebra is \textsc{NexpTime}-complete. The second result that we should mention is that the satisfiability problem for $\mathrm{GRA}(E,\neg,\Dot{\cap},\exists)$ is \textsc{Tower}-complete, since it contains $\mathrm{FL}$ and it can be translated to $\mathrm{FL}$ with equality, for which the satisfiability problem was recently proved in \cite{flutedequality} to be \textsc{Tower}-complete.

\section{Tables and normal forms}

In this article we are going to perform several model constructions and hence it is useful to start by collecting some definitions and tools that we are going to need in the later sections. We will start by defining the concept of a table which serves as an approximation of the more standard definition of a type from model theory.

\begin{definition}
    Let $k\in \mathbb{Z}_+$ and $\mathcal{F} \subseteq \{I,s\}$. A \emph{$k$-table} with respect to $\mathcal{F}$ is a maximally consistent set of $k$-ary terms of the form $\mathcal{T}$ or $\neg \mathcal{T}$, where $\mathcal{T} \in \mathrm{GRA}(\mathcal{F})$. Given a model $\mathfrak{A}$ and $\overline{a} \in A^k$, we will use $tp_\mathfrak{A}(\overline{a})$ to denote the $k$-table realized by $\overline{a}$.
\end{definition}

We will identify $k$-tables $\rho$ with the terms $\bigcap_{\alpha \in \rho} \alpha$, which makes sense since all of the algebraic signatures that we are going to consider always include the operator $\cap$. This allows us to use notation such as $\rho \models \rho'$, where $\rho$ and $\rho'$ are $k$-tables. Furthermore, we will refer to $1$-tables also as $1$-types. We say that $a\in A$ is \textbf{king}, if there is no other element in the model that realizes the same $1$-type.

Notice that there is almost no "overlapping" between tables. For instance, if we consider tables for $\varnothing$, then the table realized by a tuple $(a_1,...,a_k)$ will not imply \emph{anything} about the table realized by any non-identity permutation of the tuple $(a_1,...,a_k)$ or any sub-tuple of $(a_1,...,a_k)$. And even if we are considering tables for $\{s\}$, the table realized by $(a_1,...,a_k)$ will only imply something about the table realized by $(a_1,...,a_k,a_{k-1})$. This lack of overlapping between tables is crucial for all of the model constructions done in this article.

The purpose of a type ---typically--- is to capture all the information about a tuple that could be expressed by using a quantifier-free formula. Analogously, we would like to capture all the information about a tuple that could be expressed using a term which does not contain instances of the operators $\exists,\exists_1$ and $\exists_0$.

\begin{definition}
    Let $\mathfrak{A}$ and $\mathfrak{B}$ be models over the same vocabulary, and let $\mathcal{F} \subseteq \{I,s,E,\neg,C,\cap,\Dot{\cap}\}$. Let $\overline{a} \in A^k$ and $\overline{b} \in B^k$, where $k\in \mathbb{Z}_+$. We say that $\overline{a}$ and $\overline{b}$ are \emph{similar} with respect $\mathcal{F}$, if for every $k$-ary term $\mathcal{T} \in \mathrm{GRA}(\mathcal{F})$ we have that
    \[\overline{a} \in \llbracket \mathcal{T} \rrbracket_\mathfrak{A} \iff \overline{b} \in \llbracket \mathcal{T} \rrbracket_\mathfrak{B}.\]
\end{definition}

For different subsets of $\{I,s,E,\neg,C,\cap,\Dot{\cap}\}$ one can find explicit characterizations for when two tuples are similar using the notions of $1$-types and tables. For example, if $\mathcal{F} = \{s,C,\neg,\cap\}$, then two tuples $\overline{a}$ and $\overline{b}$ are similar with respect to $\mathcal{F}$ if and only if $tp_\mathfrak{A}(\overline{a}) = tp_\mathfrak{B}(\overline{b}), tp_\mathfrak{A}(a_{k-1}) = tp_\mathfrak{B}(b_{k-1})$ and $tp_\mathfrak{A}(a_k) = tp_\mathfrak{B}(b_k)$. The reader is encouraged to try to come up with similar characterizations for different subsets of $\mathcal{F}$, since we are going to use such characterizations heavily in the rest of this article.

A standard technique when constructing models of bounded size for formulas of $\mathrm{FO}$ is to restrict attention to formulas of very specific form. We will next introduce two such normal forms. In the normal forms we will use the operator $\cup$ which can be defined in a standard way in terms of $\neg$ and $\cap$.

\begin{definition}
    Let $\mathcal{F} \subseteq \{I,s,E,C\}$. We say that a term $\mathcal{T} \in \mathrm{GRA}(\mathcal{F} \cup \{\neg,\cap,\exists\})$ is in \emph{normal form}, if it has the following form
    \[\bigcap_{i\in I'} \exists \kappa_i \cap \bigcap_{j\in J'} \forall \lambda_j \cap \bigcap_{i\in I} \forall^{n_i}(\neg \alpha_i^\exists \cup \exists \beta_i^\exists) \cap \bigcap_{j\in J} \forall^{n_j}(\neg \alpha_j^\forall \cup \forall \beta_j^\forall),\]
    where $\kappa_i,\lambda_j,\alpha_i^\exists,\beta_i^\exists,\alpha_j$ and $\beta_j^\forall$ are terms of $\mathrm{GRA}(\mathcal{F} \cup \{\neg,\cap\})$. Here $\forall$ is short-hand notation for $\neg \exists \neg$ and $\forall^n$ stands for a sequence of $\forall$ of length $n$.
\end{definition}

\begin{definition}
    Let $\mathcal{F} \subseteq \{I,s,E,C\}$. We say that a term $\mathcal{T} \in \mathrm{GRA}(\mathcal{F} \cup \{\neg,\cap,\Dot{\cap},\exists_1,\exists_0\})$ is in \emph{normal form}, if it has the following form
    \[\bigcap_{i\in I'} \exists_0 \kappa_i \cap \bigcap_{j\in J'} \forall_0 \lambda_j \cap \bigcap_{i\in I} \forall_0 (\neg \alpha_i^\exists \cup \exists_1 \beta_i^\exists) \cap \bigcap_{j\in J} \forall_0(\neg \alpha_j^\forall \cup \forall_1 \beta_j^\forall),\]
    where $\kappa_i,\lambda_j,\alpha_i^\exists,\beta_i^\exists,\alpha_j$ and $\beta_j^\forall$ are terms of $\mathrm{GRA}(\mathcal{F} \cup \{\neg,\cap,\Dot{\cap}\})$, and the terms $\kappa_i$ and $\lambda_j$ are unary. Here $\forall_0$ and $\forall_1$ are short-hand notations for $\neg \exists_0 \neg$ and $\neg \exists_1 \neg$ respectively.
\end{definition}

In a rather standard fashion one can prove the following lemma.

\begin{lemma}
    Let $\mathcal{F} \subseteq \{I,s,E,C\}$.
    \begin{enumerate}
        \item There is a polynomial time nondeterministic procedure, taking as its input a term $\mathcal{T} \in \mathrm{GRA}(\mathcal{F} \cup \{\neg,\cap,\exists\})$ and producing a term $\mathcal{T}'$ in normal form (over extended signature), such that
        \begin{itemize}
            \item if $\mathfrak{A}\models \mathcal{T}$, for some structure $\mathfrak{A}$, then there exists a run of the procedure which produces a term $\mathcal{T}'$ in normal form so that $\mathfrak{A}' \models \mathcal{T}$ for some expansion $\mathfrak{A}'$ of $\mathfrak{A}$.
            \item if the procedure has a run producing $\mathcal{T}'$ and $\mathfrak{A} \models \mathcal{T}'$, for some $\mathfrak{A}$, then $\mathfrak{A} \models \mathcal{T}$.
        \end{itemize}
        \item There is a polynomial time nondeterministic procedure, which operates similarly as the above procedure with the exception that it takes as its input a term in $\mathcal{T} \in \mathrm{GRA}(\mathcal{F} \cup \{\neg,\cap,\Dot{\cap},\exists\})$, and which satisfies the additional requirement that if $\mathcal{T}$ does not contain the operator $\Dot{\cap}$, then neither does any of the terms that this procedure produces.
    \end{enumerate}
\end{lemma}
\begin{proof}
    We will sketch a proof for the first claim. If $\mathcal{T}$ is a $0$-ary term of the form $\exists \mathcal{P}$ or $\neg \exists \mathcal{P}$, where $\mathcal{P}$ is quantifier-free, then it already is in normal form. Otherwise $\mathcal{T}$ contains a proper subterm $\exists \mathcal{P}$, where $\mathcal{P}$ is a quantifier-free term. If $\mathcal{P}$ is unary, then we will guess a truth value, and replace $\exists \mathcal{P}$ with either $\bot$ or $\top$ according to this guess. If the resulting term is $\mathcal{T}'$, then $\mathcal{T}$ is equi-satisfiable with either $\mathcal{T}' \cap \exists \mathcal{P}$ or $\mathcal{T}' \cap \neg \exists \mathcal{P}$.
    
    In the case where the arity of $\mathcal{P}$ is at least two, we will introduce a fresh relational symbol $R$ of same arity as $\exists \mathcal{P}$, and replace the latter with the former in $\mathcal{T}$. If the resulting term is $\mathcal{T}'$, then it is easy to verify that $\mathcal{T}$ is satisfiable over the same domain as
    \[\mathcal{T}' \cap \forall^{ar(R)}(\neg R\cup \exists \mathcal{P}) \cap \forall^{ar(R)}(R\cup \forall \neg \mathcal{P}).\]
    In both cases by repeating the above process on the term $\mathcal{T}'$, we will eventually end up with an equi-satisfiable term which is in normal form.
\end{proof}

\begin{remark}
    If $C\not\in \mathcal{F}$, then the above procedures can be replaced by deterministic ones, by simply increasing the arity of each relational symbol that occurs in the input term, since then the resulting term does not contain proper $0$-ary subterms.
\end{remark}

To conclude this section, we will introduce some further notation and terminology which will be useful in the later sections of this article. Consider a term $\mathcal{T}$ in normal form. Subterms of $\mathcal{T}$ that are of the form
\[\forall^{n_i}(\neg \alpha_i^\exists \cup \exists \beta_i^\exists)\]
or
\[\forall_0(\neg \alpha_i^\exists \cup \exists_1 \beta_i^\exists)\]
are called \textbf{existential requirements} and we will denote them with $\mathcal{T}_i^\exists$. Similarly subterms of the form $\forall^{n_j}(\neg \alpha_j^\forall \cup \forall \beta_j^\forall)$ or of the form $\forall_0(\neg \alpha_j^\forall \cup \forall_1 \beta_j^\forall)$ will be called \textbf{universal requirements} and we will denote them with $\mathcal{T}_j^\forall$. 

Consider a model $\mathfrak{A}$ and an existential requirement $\mathcal{T}_i^\exists$. If $\mathcal{T}_i^\exists$ is of the form $\forall^{n_i}(\neg \alpha_i^\exists \cup \exists \beta_i^\exists)$ and $\overline{a} \in \llbracket \alpha_i^\exists \rrbracket_\mathfrak{A}$, then an element $c\in A$ so that $\overline{a}c\in \llbracket \beta_i^\exists \rrbracket_\mathfrak{A}$ will be called a \textbf{witness} for $\overline{a}$ and $\mathcal{T}_i^\exists$. Similarly, if $\mathcal{T}_i^\exists$ is of the form $\forall_0(\neg \alpha_i^\exists \cup \exists_1 \beta_i^\exists)$ and $a\in \llbracket \alpha_i^\exists \rrbracket_\mathfrak{A}$, then a tuple $\overline{c}\in A^k$, where $k = ar(\beta_i^\exists) - 1$, is called a witness for $a$ and $\mathcal{T}_i^\exists$.

\section{Ordered logic with equality}

In this section we will study the complexity of $\mathrm{GRA}(E,\neg,\cap,\exists)$, i.e. ordered logic with equality. We will start by proving that this logic has a polynomially bounded model property, which means that each satisfiable term has a model of size at most polynomial with respect to the size of the term.

Before proceeding with the proof, we will first note that w.l.o.g. we can assume that if an element $c$ is a witness for some existential requirement $\mathcal{T}_i^\exists$ and a tuple $(a_1,...,a_k)$, then $a_k \neq c$. This follows from the observation that if $\mathcal{T}_i^\exists$ is of the form
\[\forall^{n_i}(\neg \alpha_i^\exists \cup \exists \beta_i^\exists)\]
then it is equivalent with the following term
\[\forall^{n_i}(\neg (\alpha_i^\exists \cap \neg \exists E \beta_i^\exists) \cup \exists \beta_i^\exists)\]
where we can replace $\exists E \beta_i^\exists$ with $I\beta_i^\exists$. This is the exact reason why it is convenient to extend the syntax of the ordered logic with the operator $I$.

\begin{theorem}\label{linearlyorderedpolynomialmodel}
    Let $\mathcal{T} \in \mathrm{GRA}(I,E,\neg,\cap,\exists)$ and suppose that $\mathcal{T}$ is satisfiable. Then $\mathcal{T}$ has a model of size bounded polynomially in $|\mathcal{T}|$.
\end{theorem}
\begin{proof}
    Let $\mathcal{T} \in \mathrm{GRA}(I,E,\neg,\cap,\exists)$ be a term in normal form. Let $\mathfrak{A}$ be a model of $\mathcal{T}$. Without loss of generality we will assume that $\mathfrak{A}$ contains at least two distinct elements. Our goal is to construct a bounded model $\mathfrak{B} \models \mathcal{T}$.
    
    As the domain of our model we will take the following set
    \[B = I' \times I \times \{0,1\}.\]
    To define the model, we just need to specify the tables for all the $k$-tuples of elements from $B$. This will be done inductively, and in such a way that the following condition is maintained: for every $\overline{b} \in B^k$ there exists $\overline{a} \in A^k$ so that $\overline{b}$ is similar with $\overline{a}$. Maintaining this requirement will make sure that our model $\mathfrak{B}$ will not violate any universal requirements.
    
    We will start by defining the $1$-types for all the elements of $B$. Since $\mathfrak{A}\models \bigcap_{i\in I'} \exists \kappa_i$, for every $i\in I'$ there exists $a_i\in A$ so that $tp_\mathfrak{A}(a_i)\models \kappa_i$. We will define that for every $(i,i',j) \in B$, $tp_\mathfrak{B}((i,i',j)) = tp_\mathfrak{A}(a_i)$. Suppose then that we have defined the tables for $k$-tuples and we wish to define the tables for $(k+1)$-tuples. We will start by making sure that all the existential requirements are full-filled. So, let $i\in I$ and $\overline{b} \in B^k$ so that we have not assigned a witness for $\overline{b}$ and $\mathcal{T}_i^\exists$. By construction we know that there exists $\overline{a} \in A^k$ which is similar to $\overline{b}$. Now there exists $a_k \neq c\in A$ so that $\overline{a}c \in \llbracket \beta_i^\exists \rrbracket$. If $b_k = (i',i'',j)$, then we will use the element $d = (i',i,j+1 \mod 2)$ as a witness for $\overline{b}$ by defining that $tp_\mathfrak{B}(\overline{b}d) = tp_\mathfrak{A}(\overline{a}c)$. Since we have reserved for every element $|I|$ distinct witnesses for the existential requirements, the process of providing witnesses can be done without conflicts.
    
    Having provided witnesses for $k$-tuples, we will still need to do define the $(k+1)$-tables for the remaining $k$-tables. So, let $\overline{b} \in B^k$ and $d\in B$ be elements so that the table of $\overline{b}d$ has not been defined. If $b_k =d$, then the table for $\overline{b}d$ is determined by the table for $\overline{b}$. Suppose then that $b_k \neq d$. Let $\overline{a} \in A^k$ be a $k$-tuple which is similar with $\overline{b}$. Pick an arbitrary $a_k\neq c\in A$ and define $tp_\mathfrak{B}(\overline{b}d) = tp_\mathfrak{A}(\overline{a}c)$.
\end{proof}

The above theorem implies almost immediately that if we assume that the underlying vocabulary to be bounded, i.e. there is a fixed constant bound on the maximum arity of relational symbols, then the complexity of the ordered logic is quite low. However, we need to first establish that the combined complexity of the model checking problem for the logic is in \textsc{P} over bounded vocabularies.

\begin{lemma}
    Let $k \in \mathbb{Z}_+$ and suppose that $\tau$ is a vocabulary where each term has an arity of at most $k$. Let $\mathcal{T} \in \mathrm{GRA}(E,\neg,\cap,\exists)[\tau]$ be a term and suppose that $\mathfrak{A}$ is a model of the same vocabulary. Then $\llbracket \mathcal{T} \rrbracket_\mathfrak{A}$ can be calculated in time polynomial with respect to $|\mathcal{T}|\times |A|$.
\end{lemma}
\begin{proof}
    Let $\tau$ be the bounded vocabulary and $\mathfrak{A}$ be a model over $\tau$. Since none of the operators in the set $\{E,\neg,\cap,\exists\}$ increases the arity of the terms, one can use an easy induction over the structure of the terms $\mathcal{T} \in \mathrm{GRA}(E,\neg,\cap,\exists)[\tau]$ to establish that $ar(\llbracket \mathcal{T} \rrbracket_\mathfrak{A}) \leq k$. In particular the size of $\llbracket \mathcal{T} \rrbracket_\mathfrak{A}$ is bounded above by $|A|^k$. Since $k$ was a fixed constant, this implies that the size of the interpretation of each term over $\mathfrak{A}$ is polynomially bounded with respect to the size of the model. Hence it is clear that the interpretation $\llbracket \mathcal{T} \rrbracket_\mathfrak{A}$ can be calculated in time polynomial with respect to $|\mathcal{T}| \times |A|$.
\end{proof}

\begin{theorem}
    The satisfiability problem for $\mathrm{GRA}(E,\neg,\cap,\exists)$ over bounded vocabularies is \textsc{NP}-complete.
\end{theorem}
\begin{proof}
    Let's start with the upper bound. Given a term $\mathcal{T}$ in normal form, a non-deterministic algorithm can simply guess a model of size polynomial with respect to $|\mathcal{T}|$ and check that it is a model of $\mathcal{T}$. Note that our assumption guarantees that not only is the size of the model polynomially bounded with respect to $|\mathcal{T}|$ but also the description of the model, which together with the previous lemma implies that this algorithm runs in polynomial time. For the lower bound, one can reduce boolean satisfiability problem to the satisfiability problem of $\mathrm{GRA}(\neg,\cap,\exists)$ over unary vocabularies.
\end{proof}

In the case where the vocabulary is not assumed to be bounded, the complexity of the ordered logic turns out to be \textsc{Pspace}-complete. To prove hardness, we will reduce the satisfiability problem of modal logic over serial frames to that of $\mathrm{GRA}(\neg,\cap,\exists)$.

\begin{lemma}
    The satisfiability problem for $\mathrm{GRA}(\neg,\cap,\exists)$ is \textsc{Pspace}-hard.
\end{lemma}
\begin{proof}
    We will give a reduction from the satisfiability problem of modal logic over serial frames. Fix a sentence $\varphi \in \mathrm{ML}$. Let $\{p_1,...,p_n\}$ be the set of propositional symbols occurring in $\varphi$ and let $d$ be the modal depth of $\varphi$. For every $1\leq i\leq n$ and $1\leq k\leq d + 1$ we will reserve a $k$-ary relational symbol $R_{i,k}$. Next we will define a family of translations $(t_k)_{1\leq k\leq d + 1}:\mathrm{Subf}(\varphi) \to \mathrm{GRA}(\neg,\cap,\exists)$ recursively as follows:
    \begin{enumerate}
        \item $t_k(p_i) = R_{i,k}$
        \item $t_k(\neg \psi) = \neg t_k(\psi)$
        \item $t_k(\psi \land \chi) = t_k(\psi) \cap t_k(\chi)$
        \item $t_k(\lozenge \psi) = \exists t_{k+1}(\psi)$
    \end{enumerate}
    Our goal is to prove now that $\exists t_1(\varphi)$ is satisfiable if and only if $\varphi$ is.
    
    Suppose first that $\exists t_1(\varphi)$ is satisfiable. Let $\mathfrak{A}$ be a model and let $a\in A$ be an element so that $a\in \llbracket t_1(\varphi) \rrbracket_\mathfrak{A}$. We can now define a Kripke model $\mathfrak{M} = (W,S,V)$ as follows.
    \begin{enumerate}
        \item $W = \{\overline{a} \mid \overline{a} \in A^{<\omega}\}$.
        \item $S = \{(\overline{a},\overline{a}b) \mid \overline{a}b \in A^{<\omega}\}$.
        \item For every $p_i$ we define
        \[V(p_i) = \bigcup_{1\leq k \leq d + 1} R_{i,k}^\mathfrak{A}\]
    \end{enumerate}
    Clearly $\mathfrak{M}$ is a serial Kripke model. Using a straightforward induction, one can prove that for every $\psi \in \mathrm{Subf}(\varphi)$ of modal depth $k$ and for every tuple $\overline{b} \in A^k$ we have that $\mathfrak{M},\overline{b} \vdash \psi$ if and only if $\overline{b} \in \llbracket t_k(\psi) \rrbracket_\mathfrak{A}$. In particular we have that $\mathfrak{M},a\vdash \varphi$ iff $a \in \llbracket t_1(\psi) \rrbracket_\mathfrak{A}$, which - together with our assumption $a\in \llbracket t_1(\varphi) \rrbracket_\mathfrak{A}$ - implies that $\mathfrak{M},a \vdash \varphi$.
    
    Suppose then that $\varphi$ is satisfiable. Let $\mathfrak{M} = (W,S,V)$ be a Kripke model and let $w_0\in W$ be a world for which $\mathfrak{M},w_0 \vdash \varphi$. We define $\mathcal{F}$ to be the following set
    \[\mathcal{F} = \{f:W\to W \mid f \subseteq S\}.\]
    Note that since $S$ is serial, for every $(w,w') \in S$ there exists $f\in \mathcal{F}$ with the property that $f(w) = w'$. Now fix some function $g\in \mathcal{F}$. We define a model $\mathfrak{A}$ as follows.
    \begin{enumerate}
        \item $A = \mathcal{F}$
        \item For every $1\leq i\leq n$ and $1\leq k \leq d + 1$ we define
        \[R_{i,k}^\mathfrak{A} = \{(g,f_1,...,f_{k-1}) \in A^k \mid (f_{k-1} \circ ... \circ f_1)(w_0) \in V(p_i)\}\]
        Here we agree that if $k = 1$, then $(f_{k-1} \circ ... \circ f_1)(w_0) = w_0$. Thus for every $1\leq i \leq n$ we have that $g \in R_{i,1}^\mathfrak{A}$ if and only if $w_0 \in V(p_i)$.
    \end{enumerate}
    Using a straightforward induction, one can prove that for every $\psi \in \mathrm{Subf}(\varphi)$ of modal depth $k$ and for every $(f_1,...,f_{k-1}) \in A^{k-1}$ we have that $\mathfrak{M},(f_{k-1} \circ ... \circ f_1)(w_0) \vdash \psi$ if and only if $(g,f_1,...,f_{k-1}) \in \llbracket t_k(\psi) \rrbracket_\mathfrak{A}$, which again implies the desired claim.
\end{proof}

The rest of this section will be devoted towards proving the corresponding upper bound for the logic $\mathrm{GRA}(E,\neg,\cap,\exists)$. The basic idea is to use a variant of the well-known Lardner's algorithm, which tries to construct a model $\mathfrak{A}$ for a given term in a depth-first fashion.

The algorithm will start by converting the given term to an equi-satisfiable term in normal form. If the term has a model of size one, then the algorithm accepts. Otherwise the algorithm will initialize a set $A$ of size $2|I'||I|$, which will serve as the domain of the model $\mathfrak{A}$. For every $a\in A$, the algorithm guesses a $1$-type $\pi$ and then assigns it to that element. The algorithm will reject if there exists $i\in I'$ so that for no $a\in A$ it is the case that $tp_\mathfrak{A}(a) \models \kappa_i$. Similarly the algorithm will reject if there is $j\in J'$ and $a\in A$ so that $tp_\mathfrak{A}(a) \not\models \lambda_j$.

After this the algorithm will begin to recursively guess tables for extensions of tuples of elements from $A$. At each stage the algorithm needs to guess extensions for some tuple $\overline{a}\in A^k$. The algorithm will star by guessing witnesses for the existential requirements $\mathcal{T}_i^\exists$ for which $\overline{a} \in \llbracket \alpha_i^\exists\rrbracket_\mathfrak{A}$, which will be done as follows. First the algorithm guesses a set $C\subseteq A$ of size $|I|$ so that $a_k \not\in C$. Then, for every existential requirement $\mathcal{T}_i^\exists$ the algorithm guesses an element $c\in C$ and a table $\rho_i$, and defines the table of $\overline{a}c$ to be $\rho_i$. The algorithm then checks that for every existential requirement $\mathcal{T}_i^\exists$ for which $tp_\mathfrak{A}(\overline{a}) \models \alpha_i^\exists$, we should have that $\rho_i \models \beta_i^\exists$, and if this is not the case then it rejects. After providing existential witnesses, the algorithm will guess tables for the remaining extensions of $\overline{a}$, with the exception of $\overline{a}a_k$, since its table is determined by the table of $\overline{a}$. Now if there is a universal requirement $\mathcal{T}_j^\forall$ and $c\in A$ so that $tp_\mathfrak{A}(\overline{a}) \models \alpha_j^\forall$ but $tp_\mathfrak{A}(\overline{a}c) \not\models \beta_j^\forall$, then the algorithm rejects. Otherwise the algorithm will perform the above operation recursively on the extensions of $\overline{a}$.

Let us then briefly verify that the algorithm only uses polynomial space. First we note that checking whether a given term has a model of size one can be done easily using only polynomial space. Secondly, when guessing the tables for the extensions of a tuple $\overline{a} \in A^k$, the algorithm needs to only remember the set $A$ and its previous guesses for the tables of extensions of the prefixes of $\overline{a}$. Since the size of $A$ is polynomial and the length of $\overline{a}$ is polynomially bounded (since we need to specify tables only for tuples with length at most the maximum arity of a relational symbol that occurs in $\mathcal{T}$), all of these require only polynomial space.

\begin{lemma}
    The above algorithm accepts iff the input term $\mathcal{T}$ is satisfiable.
\end{lemma}
\begin{proof}
    Suppose first that $\mathcal{T}$ has a model $\mathfrak{A}$. If the size of this model is one, then the algorithm clearly accepts. On the other hand, if $\mathfrak{A}$ has size at least two, then the proof of theorem \ref{linearlyorderedpolynomialmodel} shows that we can assume it to have a domain of size $2|I'||I|$. Then it is clear that there exists an accepting run for the algorithm, which consists of guessing tables according to the way they are defined in the model $\mathfrak{A}$.
    
    Suppose then that the algorithm accepts the term. If the algorithm accepts because there exists a model of size one, then the claim is clear. In the other case we note that since the algorithm assigns tables for each relevant tuple of elements of $A$, and there are no tuples which receive different tables, we can clearly read off a model from an accepting run of the algorithm.
\end{proof}

Since \textsc{NPspace} = \textsc{Pspace}, the following corollary is immediate.

\begin{corollary}
    The satisfiability problem for $\mathrm{GRA}(E,\neg,\cap,\exists)$ is \textsc{Pspace}-complete.
\end{corollary}

\section{Further extensions of ordered logic}

In this section we will study extensions of ordered logic which are obtained by adding either the swap or the one-dimensional intersection (or both) into its syntax. It turns out that we can deduce easily from the literature sharp lower bounds for the relevant fragments, while the upper bounds require some work.

It follows rather directly from the literature that the satisfiability problem for $\mathrm{GRA}(\neg,C,\cap,\exists)$ is \textsc{NexpTime}-hard already over vocabularies with at most binary relational symbols. For instance, one can easily reduce the satisfiability problem for \emph{boolean modal logic} to that of $\mathrm{GRA}(\neg,C,\cap,\exists)$. Since the former is \textsc{NexpTime}-hard \cite{Lutz2000TheCO}, the following claim is immediate.

\begin{theorem}\label{chardness}
    The satisfiability problem is \textsc{NexpTime}-hard for $\mathrm{GRA}(\neg,C,\cap,\exists)$.
\end{theorem}

Concerning the addition of the swap operator, it is again not difficult to show that the satisfiability problem for $\mathrm{GRA}(s,\neg,\cap,\exists)$ is \textsc{NexpTime}-hard. Here we will prove it by reducing the satisfiability problem of $\mathrm{S5}^2$ to that of $\mathrm{GRA}(s,\neg,\cap,\exists)$, which was proved to be \textsc{NexpTime}-hard in \cite{complexityS52}.

\begin{theorem}\label{swaphardness}
    The satisfiability problem is \textsc{NexpTime}-hard for $\mathrm{GRA}(s,\neg,\cap,\exists)$.
\end{theorem}
\begin{proof}
    Let $\varphi \in \mathrm{S5}^2$. Our goal is to construct a term $\mathcal{T} \in \mathrm{GRA}(s,\neg,\cap,\exists)$ which is satisfiable if and only if $\varphi$ is. We will begin by fixing the set of relational symbols occurring in the term. First, for every propositional symbol $p$ occurring in $\varphi$, we will add a binary relational symbol $P$. Secondly, for each subformula $\varphi$ of the form $\lozenge_i \psi$, we will add a binary relational symbol $S_{\lozenge_i \psi}$. Let $\tau$ denote the resulting vocabulary. Note that $\tau$ consist of relational symbols of arity at most two.
    
    Now we are ready to define recursively a mapping 
    \[t:\mathrm{Subf}(\varphi) \rightarrow \mathrm{GRA}(s,\neg,\cap,\exists)[\tau],\]
    where $\mathrm{Subf}(\varphi)$ denotes the set of subformulas of $\varphi$, as follows.
    \begin{enumerate}
        \item $t(p) = P$, where $p$ is a propositional symbol.
        \item $t(\neg \psi) = \neg t(\psi)$.
        \item $t(\psi \land \chi) = t(\psi) \cap t(\chi)$.
        \item $t(\lozenge_i \psi) = S_{\lozenge_i \psi}$.
    \end{enumerate}
    Then we define
    \[\mathcal{T} := \exists\exists t(\varphi) \cap \bigcap_{\lozenge_1 \psi \in \mathrm{Subf}(\varphi)} \forall ((\neg \exists S_{\lozenge_1 \psi} \cup \exists t(\psi)) \cap (\neg \exists t(\psi) \cup \forall S_{\lozenge_1 \psi}))\]
    \[\cap \bigcap_{\lozenge_2 \psi \in \mathrm{Subf}(\varphi)} \forall ((\neg \exists s S_{\lozenge_2 \psi} \cup \exists s t(\psi)) \cap (\neg \exists s t(\psi) \cup \forall s S_{\lozenge_2 \psi})).\]
    It is straightforward to check that $\varphi$ is satisfiable if and only if $\mathcal{T}$ is.
\end{proof}

\begin{remark}
    It would have been also possible to reduce the satisfiability problem for $\mathrm{GRA}(\neg,C,\cap,\exists)$ to that of $\mathrm{GRA}(s,\neg,\cap,\exists)$. The reason why we chose to prove the \textsc{NexpTime}-hardness via a reduction from $\mathrm{S5}^2$ is that later we can use a very similar translation from $\mathrm{S5}^3$ to prove that the extension of ordered logic with cyclic permutation is undecidable.
\end{remark}

Next we will prove the corresponding upper bounds on the complexities of $\mathrm{GRA}(\neg,C,\cap,\exists)$ and $\mathrm{GRA}(s,\neg,\cap,\exists)$ by showing that their least common extension $\mathrm{GRA}(s,\neg,C,\cap,\exists)$ has the exponentially bounded model property. The core of the argument is the same as the proof of theorem \ref{linearlyorderedpolynomialmodel}, although the presence of $C$ and $s$ forces us to be a bit more careful with the way we provide witnesses for existential requirements.

\begin{theorem}\label{swaprefuniformboundedmodel}
    Let $\mathcal{T} \in \mathrm{GRA}(s,\neg,C,\cap,\exists)$ and suppose that $\mathcal{T}$ is satisfiable. Then $\mathcal{T}$ has a model of size bounded exponentially in $|\mathcal{T}|$.
\end{theorem}
\begin{proof}
    Let $\mathcal{T} \in \mathrm{GRA}(s,\neg,C,\cap,\exists)$ be a term in normal form and let $\mathfrak{A}$ be a model of $\mathcal{T}$. Our goal is to construct a bounded model $\mathfrak{B}$ so that $\mathfrak{B} \models \mathcal{T}$. As the domain of the model $\mathfrak{B}$ we will take the set
    \[B := \{tp_\mathfrak{A}(a) \mid a \in A\} \times I \times \{0,1,2\}.\]
    Clearly $|B| \leq 2^{O(|\mathcal{T}|)}$. Again, to construct the model, we will need to specify the tables for all the $k$-tuples of elements from $B$. We will follow the same strategy as in the proof of theorem \ref{linearlyorderedpolynomialmodel}, i.e. the tables will be specified inductively while maintaining the condition that for every $\overline{b}\in B^k$ for which $tp_\mathfrak{B}(\overline{b})$ has been specified, there exists $\overline{a} \in A^k$ which is similar to $\overline{b}$.
    
    We will start with the $1$-types. For every $b = (tp_\mathfrak{A}(a),i,j) \in B$ we define that $tp_\mathfrak{B}(b) := tp_\mathfrak{A}(a)$. Suppose then that we have defined the tables for $k$-tuples. We start defining the tables for $(k+1)$-tuples by providing witnesses for all the relevant tuples. So, consider an existential requirement $\mathcal{T}_i^\exists$ and a tuple $\overline{b}\in B^k$ so that $\overline{b}\in \llbracket \alpha_i \rrbracket_\mathfrak{B}$. Suppose that $b_k = (tp_\mathfrak{A}(a),i',j)$. By construction there exists a tuple $\overline{a}\in A^k$ so that $\overline{b}$ and $\overline{a}$ are similar. Thus $\overline{a}\in \llbracket \alpha_i\rrbracket_\mathfrak{A}$. Since $\mathfrak{A}\models \mathcal{T}_i^\exists$, there exists an element $c\in A$ which is a witness for $\overline{a}$ and $\mathcal{T}_i^\exists$. We will use the element $d = (tp_\mathfrak{A}(c),i,j+1 \mod 3) \in B$ as a witness for $\overline{b}$ by defining that $tp_\mathfrak{B}(\overline{b}d) := tp_\mathfrak{A}(\overline{a}c)$ and $tp_\mathfrak{B}((b_1,...,b_{k-1},d,b_k)) := tp_\mathfrak{A}((a_1,...,a_{k-1},c,a_k))$.
    
    Before moving forward, let us argue that our method of assigning witnesses does not produce conflicts. So consider a tuple $\overline{b} = (b_1,...,b_k) \in B^k$ and $d\in B$ so that we used $d$ as a witness for $\overline{b}$ and some existential requirement $\mathcal{T}_i^\exists$. We will argue that the table for the tuple $\overline{b}d$ was not defined in two different ways. First we note that we have reserved distinct elements for each of the existential requirements, and thus we used $d$ as a witness for $\overline{b}$ only for the existential requirement $\mathcal{T}_i^\exists$. We then note that since we are assigning witnesses for tuples in a "cyclic" manner, we will not use $b_k$ as a witness for the tuple $(b_1,...,b_{k-1},d)$. Since these cases are the only possible ways that we might have defined the table of the tuple $\overline{b}d$ in two different ways, we conclude that it is only defined once.
    
    We will now assign tables for the remaining $(k+1)$-tuples. So, consider a tuple $\overline{b}\in B^k$ and $d = (tp_\mathfrak{A}(c),i,j)\in B$ so that we have not defined the table for the tuple $\overline{b}d$. By construction there exists a tuple $\overline{a}\in A^k$ which is similar to $\overline{b}$. Let $c\in A$ be an element that realizes the $1$-type of $d$ (and which is not necessarily distinct from $a_k$). Now we define that $tp_\mathfrak{B}(\overline{b}d) := tp_\mathfrak{A}(\overline{a}c)$ and that $tp_\mathfrak{B}((b_1,...,b_{k-1},d,b_k)) = tp_\mathfrak{A}((a_1,...,a_{k-1},c,a_k))$.
\end{proof}

\begin{corollary}
    The satisfiability problem for $\mathrm{GRA}(s,\neg,C,\cap,\exists)$ is \textsc{NexpTime}-complete.
\end{corollary}

As the final results of this section we will consider the logics $\mathrm{GRA}(E,\neg,C,\cap,\exists)$ and $\mathrm{GRA}(s,E,\neg,C,\cap,\exists)$. An easy modification in the argument of theorem \ref{linearlyorderedpolynomialmodel} yields a bounded model property for the first logic.

\begin{theorem}
    Let $\mathcal{T} \in \mathrm{GRA}(E,\neg,C,\cap,\exists)$ and suppose that $\mathcal{T}$ is satisfiable. Then $\mathcal{T}$ has a model of size bounded exponentially in $|\mathcal{T}|$.
\end{theorem}
\begin{proof}
    Let $\mathcal{T}$ be a term in normal form and assume that $\mathfrak{A}$ is a model of $\mathcal{T}$. As the domain of the bounded model $\mathfrak{B}$ one can now take the set
    \[B := K \cup \{tp_\mathfrak{A}(a) \mid \text{$a$ is not a king}\} \times I \times \{0,1\},\]
    where $K = \{tp_\mathfrak{A}(a) \mid \text{$a$ is a king}\}$. One can now adapt the construction in the proof of theorem \ref{linearlyorderedpolynomialmodel} to obtain a model $\mathfrak{B}$ of $\mathcal{T}$ with domain $B$. Notice that there is no need for a "court", since the table of a tuple $(a_1,...,a_{k-1},a_k)$ does not imply anything about the table of the tuple $(a_1,...,a_k,a_{k-1})$.
\end{proof}

\begin{corollary}
    The satisfiability problem for $\mathrm{GRA}(E,\neg,C,\cap,\exists)$ is \textsc{NexpTime}-complete.
\end{corollary}

The logic $\mathrm{GRA}(s,E,\neg,C,\cap,\exists)$ turns out to be more tricky. While we have not been able to verify whether this logic is undecidable, we can show that it does not have the finite model property. For instance, it is straightforward to verify that the following term $\mathcal{T}$ is an infinity-axiom.
\[\exists Z \cap \forall \exists C(S,\neg Z) \cap \forall \forall (\neg s S \cup \forall s (G \cup F)) \cap \forall \forall (\neg \exists F \cup \neg s S),\]
where $Z$ is a unary relation, $S$ is a binary relation and $F$ is ternary relation, and $G$ is a shorthand for $E(F\cup \neg F)$. For the readers convenience, we have written the above term also using a more standard syntax
\[\exists v_1 Z(v_1) \land \forall v_1 \exists v_2 (S(v_1,v_2) \land \neg Z(v_2))\]
\[\land \forall v_1 \forall v_3 (\neg S(v_3,v_1) \lor \forall v_2 (v_2 = v_3 \lor F(v_1,v_2,v_3)))\]
\[\land \forall v_1 \forall v_2 (\neg \exists v_3 F(v_1,v_2,v_3) \lor \neg S(v_2,v_1)),\]
where we have replaced $(G\cup F)$ with the equivalent formula $v_2 = v_3 \lor F(v_1,v_2,v_3)$.

\begin{proposition}
    Suppose that $\mathfrak{A} \models \mathcal{T}$. Then $|A|\geq \omega$.
\end{proposition}
\begin{proof}
    Let $\mathfrak{A}$ be a model which satisfies $\mathcal{T}$. Thus there exists a function $f:A \rightarrow (A\backslash Z^\mathfrak{A})$ with the property that $(a,f(a)) \in S^\mathfrak{A}$, for every $a\in A$. Since $Z^\mathfrak{A} \neq \varnothing$, it suffices to show that $f$ must be an injection. So, suppose that $a\neq b$ and $f(a) = f(b)$. Now $(a,f(a)) \in S^\mathfrak{A}$, and hence $(f(a),b,a) \in F^\mathfrak{A}$, since $a\neq b$. Thus $(b,f(a)) \not\in S^\mathfrak{A}$, which is a contradiction since $f(a) = f(b)$.
\end{proof}

In $\mathcal{T}$ we used $C$ only in the term $\forall \exists C(S,\neg Z)$. This could be replaced by the term
\[\forall \exists (S \cap \neg Z') \cap \forall (\neg \exists s Z' \cup Z),\]
where $Z'$ is a fresh binary relational symbol. Thus also $\mathrm{GRA}(s,E,\neg,\cap,\exists)$ lacks finite model property.

\section{One-dimensional ordered logics}

In this section we consider logics that are obtained from the ordered logic and the fluted logic by imposing the restriction of one-dimensionality. We will start by considering the one-dimensional fluted logic extended which has been extended with the operators $s$ and $E$. As the first result of this section we will show that the satisfiability problem for this logic is \textsc{NexpTime}-complete. As usual, we will prove this by showing that the logic has the bounded model property. The proof was heavily influenced by a similar model constructions performed in \cite{U1b} and \cite{Kieronski2019OnedimensionalGF}, which were originally influenced by the classical model construction used in \cite{GKV97}.

\begin{theorem}\label{BoundedModelFL1}
    Let $\mathcal{T} \in \mathrm{GRA}(s,E,\neg,\Dot{\cap},\exists_1,\exists_0)$ and suppose that $\mathcal{T}$ is satisfiable. Then $\mathcal{T}$ has a model of size bounded exponentially in $|\mathcal{T}|$.
\end{theorem}
\begin{proof}
    Let $\mathcal{T} \in \mathrm{GRA}(s,E,\neg,\Dot{\cap},\exists_1,\exists_0)$ be a term in normal form and suppose that $\mathcal{T}$ is satisfiable. Let us fix an arbitrary model $\mathfrak{A}$ of $\mathcal{T}$, which we will use to construct a bounded model $\mathfrak{B}$ for $\mathcal{T}$. Let $K \subseteq A$ denote the set of kings of $A$. For each existential requirement $\mathcal{T}_i^\exists$ of $\mathcal{T}$ and $k\in \llbracket \alpha_i^\exists \rrbracket_\mathfrak{A} \cap K$, we will pick some witness $\overline{c}$. Let $C$ denote the resulting set. Next, we let $P$ denote the set of non-royal $1$-types realized by elements of $\mathfrak{A}$. Fix some function $f:P \to A$ with the property that $tp_\mathfrak{A}(f(\pi)) = \pi$, for every $\pi \in P$. For every existential requirement $\mathcal{T}_i^\exists$ and $\pi \in P$ so that $\pi \models \alpha_i^\exists$, we will pick some witness $\overline{c}^{\pi,i}$. Let $W_{\pi,i}$ denote the set of elements occurring in $\overline{c}^{\pi,i}$ that are not kings. 
    
    As the domain of the bounded model $\mathfrak{B}$ we will then take the following set
    \[B = C \cup \bigcup_{\pi,i,j} W_{\pi,i,j},\]
    where $\pi$ ranges over $P$, $i$ ranges over $I$ and $j$ ranges over $\{0,1,2\}$. The sets $W_{\pi,i,j}$ are pairwise disjoint copies of the sets $W_{\pi,i}$. Clearly $|B| \leq 2^{O(|\mathcal{T}|)}$. We will make $\mathfrak{B} \upharpoonright C$ isomorphic with $\mathfrak{A} \upharpoonright C$. Furthermore, we will make each of the structures $\mathfrak{B} \upharpoonright (K \cup W_{\pi,i,j})$ isomorphic with the corresponding structures $\mathfrak{A} \upharpoonright (K \cup W_{\pi,i})$.
    
    We will then provide witnesses for elements of $B$. Since we have already provided witnesses for kings, we need to only provide witnesses for non-royal elements of the court and for elements in $(B\backslash C)$. We will start with the non-royal elements of the court. Consider an existential requirement $\mathcal{T}_i^\exists$ and let $b\in (C\backslash K) \cap \llbracket \alpha_i^\exists \rrbracket_\mathfrak{A}$. If there exists a witness for $b$ and $\mathcal{T}_i^\exists$ in $C$, then nothing needs to be done. So suppose that there does not exists a witness for $b$ and $\mathcal{T}_i^\exists$ in $C$. If $\pi$ is the $1$-type of $b$ in $\mathfrak{B}$, then we know that there exists a witness $\overline{c}$ for $f(\pi)$ and $\mathcal{T}_i^\exists$. We have now two cases.
    
    Suppose first that the length of $\overline{c}$ is one, i.e. $\overline{c} = c$, for some $c\in A$. If $c = a$, then $b$ is already a witness for itself in $\mathfrak{B}$. If $c\neq a$, then we define $tp_\mathfrak{B}(b,d) = tp_\mathfrak{B}(a,c)$ and $tp_\mathfrak{B}(d,b) = tp_\mathfrak{B}(c,a)$, where $d$ denotes the single element of $W_{\pi,i,0}$ (note that $d$ can't be a king, since otherwise $b$ and $\mathcal{T}_i^\exists$ would have had a witness in $C$).
    
    Suppose then that the length of $\overline{c}$ is $k > 1$. If $\overline{d} \in (W_{\pi,i,0} \cup K)^k$ denotes the corresponding witness, then we define $tp_\mathfrak{B}(b\overline{d}) = tp_\mathfrak{A}(a\overline{c})$ and $tp_\mathfrak{B}(bd_1,...,d_k,d_{k-1}) = tp_\mathfrak{B}(ac_1,...,c_k,c_{k-1})$. Note that since $b$ does not occur in $\overline{d}$ and $\overline{d}$ contains at least one non-royal element, the above definitions do not lead into any conflicts with the structure that we have assigned for $\mathfrak{B} \upharpoonright C$.
    
    Thus we have managed to provide witnesses for elements in $C\backslash K$. To provide witnesses for elements of $(B\backslash C)$, we can do roughly the same as above with the exception that instead of $W_{\pi,i,0}$, we will use - assuming that $b\in W_{\pi',i',j}$ - the set $W_{\pi,i,j+1\mod 3}$. Let us then briefly argue that the above procedure for producing witnesses can be executed without conflicts. First we note that we do not face any conflicts when assigning witnesses for some $b$ and $\mathcal{T}_i^\exists$ and then for $b$ and $\mathcal{T}_i^\exists$, where $i\neq i'$, since for every $j$ the sets $W_{\pi,i,j}$ and $W_{\pi,i',j}$ are disjoint. Secondly we note that we do not face any conflicts when assigning witnesses for some $b$ and $\mathcal{T}_i^\exists$ and then for $b\neq b'$ and $\mathcal{T}_i^\exists$, since in the first case we assign a table for the tuple $b\overline{d}$ and in the second case for $b'\overline{d}$, and neither of these tables imply anything about the other table. Finally we note that since we are assigning witnesses in a cyclic manner, if we use $\overline{d}$ as a witness for $b \not\in C$ and $\mathcal{T}_i^\exists$, then we we are never using any tuple containing $b$ as a witnesses for any of the elements in $\overline{d}$.
    
    To complete the structure, for every $k$ we need to define the tables for tuples $\overline{b} \in B^k$. We can do this inductively with respect to $k$ as follows. Suppose first that there exists distinct elements $b, b' \in B$ so that we have not assigned table for the pair $(b,b')$. Now we choose a pair of distinct elements $a,a' \in A$ with the same $1$-types as $b$ and $b'$, and then define $tp_\mathfrak{B}(b,b') = tp_\mathfrak{A}(a,a')$ and $tp_\mathfrak{B}(b',b) = tp_\mathfrak{A}(a',a)$. Note that such elements $a,a'$ exists even if the elements $b,b'$ would have the same $1$-types, since at least one of them is not a king. Suppose then that we have defined the tables for every $\overline{d} \in B^k$. Let $b\in B$ and $\overline{d} \in B^k$ be so that we have not defined the table for the tuple $(b,\overline{d})$. By construction there exists $\overline{c} \in A^k$ which is similar with $\overline{d}$. Let $a \in A$ be an arbitrary element which has the same $1$-type as $b$. We then define $tp_\mathfrak{B}(b,\overline{d}) = tp_\mathfrak{A}(a,\overline{c})$ and $tp_\mathfrak{B}(b,d_1,...,d_k,d_{k-1}) = tp_\mathfrak{A}(a,c_1,...,c_k,c_{k-1})$. Continuing this way it is clear that we can define tables for all the tuples of $B^k$ in such a way that we do not violate any of the universal requirements.
\end{proof}

\begin{corollary}
    The satisfiability problem for $\mathrm{GRA}(s,E,\neg,\Dot{\cap},\exists_1,\exists_0)$ is \textsc{NexpTime}-complete.
\end{corollary}

Next we will consider the one-dimensional ordered logic extended with the equality operator $E$, which is the logic $\mathrm{GRA}(E,\neg,\cap,\exists_1,\exists_0)$. Perhaps surprisingly, the complexity of this logic drops down from \textsc{Pspace} to \textsc{NP}.

Let $\mathcal{T} \in \mathrm{GRA}(E,\neg,\cap,\exists_1,\exists_0)$ be a term in normal form. Without loss of generality we can assume that for every universal requirement
\[\forall_0 (\neg \alpha_j^\forall \cup \forall_1 \beta_j^\forall)\]
there exists a corresponding existential requirement
\[\forall_0 (\neg \alpha_j^\forall \cup \exists_1 \beta_j^\forall).\]
Now let us describe a non-deterministic algorithm for checking whether $\mathcal{T}$ is satisfiable. The algorithm will start by initializing a set $A$ of size $2|I'||I|$ and guessing $1$-types for each element of $A$. If these $1$-types are already conflicting with $\mathcal{T}$, the algorithm rejects. Otherwise the algorithm guesses, for every $a\in A$ and existential requirement $\mathcal{T}_i^\exists$, for which the $1$-type of $a$ implies $\alpha_i^\exists$, a tuple of elements $\overline{c}$ from $A$ and a table $\rho_i$ and then defines the table of $a\overline{c}$ to be $\rho_i$, after which the algorithm checks that $\rho_i \models \beta_i^\exists$, and that there is no universal requirement $\mathcal{T}_j^\forall$ so that the $1$-type of $a$ implies $\alpha_j^\forall$ but $\rho_i \not\models \beta_j^\forall$. If the algorithm manages to assign witnesses for all of the elements in this manner, then the algorithm accepts, and otherwise it rejects.

The above algorithm clearly runs in polynomial time, since it just guesses witnesses for polynomially many elements. Perhaps surprisingly, the algorithm is also correct, even though it never constructs the whole model. This is explained by the fact that the tables that we are assigning for witnesses can be also used to perform the completion of the underlying model, which is guaranteed by the fact that for each universal requirement there exists a corresponding existential requirement.

\begin{lemma}
    The above algorithm accepts iff the input term $\mathcal{T}$ is satisfiable.
\end{lemma}
\begin{proof}
    Suppose first that $\mathcal{T}$ is satisfiable. An easy adaptation of the proof of theorem \ref{linearlyorderedpolynomialmodel} shows that $\mathcal{T}$ has a model $\mathfrak{A}$ of size $2|I'||I|$. Then it follows immediately that the above algorithm accepts, since it has at least one accepting run, namely the run where it makes it guesses according to the model $\mathfrak{A}$.
    
    Suppose then that there exists an accepting run for the above algorithm and let $\mathfrak{A}$ be the resulting structure. Now it might be the case that for some universal requirement $\mathcal{T}_j^\forall$ there exists $a\in \llbracket \alpha_j^\forall \rrbracket_\mathfrak{A}$ and $\overline{c} \in A^k$, where $k = ar(\beta_j^\forall)$, so that $tp_\mathfrak{A}(a\overline{c})$ has not been defined. By our assumption there exists a corresponding existential requirement in $\mathcal{T}$. Hence there must exists some tuple $\overline{c}'$ so that the table $\rho$ that the algorithm assigned for $a\overline{c}$ satisfies $\rho \models \beta_j^\forall$. Thus we can define $tp_\mathfrak{A}(a\overline{c}) = \rho$. Continuing this way it is possible to complete the structure $\mathfrak{A}$ into a proper model of $\mathcal{T}$.
\end{proof}

\begin{corollary}
    The satisfiability problem for $\mathrm{GRA}(E,\neg,\cap,\exists_1,\exists_0)$ is \textsc{NP}-complete.
\end{corollary}

Note that the satisfiability problem remains \textsc{NexpTime}-hard for the logics $\mathrm{GRA}(\neg,C,\cap,\exists_1,\exists_0)$ and $\mathrm{GRA}(s,\neg,\cap,\exists_1,\exists_0)$, since in the proofs of theorems \ref{chardness} and \ref{swaphardness} we used only vocabularies with at most binary relational symbols.

\section{Undecidable extensions}

In this section we will prove that several natural ordered fragments of first-order logic are undecidable. We will start by observing that adding the cyclic permutation operator to ordered logic leads to an undecidable logic.

\begin{theorem}
    The satisfiability problem for $\mathrm{GRA}(p,\neg,\cap,\exists)$ is $\Pi_1^0$-complete.
\end{theorem}
\begin{proof}
    A straightforward adaption of the proof of theorem \ref{swaphardness} allows one to reduce the satisfiability problem of $\mathrm{S5}^3$ to that of $\mathrm{GRA}(p,\neg,\cap,\exists)$, since using $p$ and $\exists$ we can project elements from a tuple in an arbitrary order. Since the satisfiability problem is $\Pi_1^0$-hard for $\mathrm{S5}^3$ and $\mathrm{GRA}(p,\neg,\cap,\exists)$ is a fragment of $\mathrm{FO}$, the claim follows.
\end{proof}

For the remaining fragments we will prove undecidability by reducing the tiling problem to their satisfiability problem. Let us start by recalling the tiling problem of the infinite grid $\mathbb{N} \times \mathbb{N}$. A tile is a mapping $t:\{R,L,T,B\} \to C$, where $C$ is a countably infinite set of colors. We let $t_X$ denote $t(X)$. Intuitively, $t_R,t_L,t_T$ and $t_B$ correspond  to the colors of the right, left, top and bottom edges of a tile. Now, let $\mathbb{T}$ be a finite set of tiles. A $\mathbb{T}$-tiling of $\mathbb{N} \times \mathbb{N}$ is a function $f:\mathbb{N} \times \mathbb{N} \to \mathbb{T}$ such that for all $i,j \in \mathbb{N}$, we have that $t_R = t_L'$ when $f(i,j) = t$ and $f(i+1,j) = t'$, and similarly, $t_T = t_B'$ when $f(i,j) = t$ and $f(i,j+1) = t'$. The tiling problem for the grid $\mathbb{N} \times \mathbb{N}$ asks, with the input of a finite set of $\mathbb{T}$ of tiles, if there exists a $\mathbb{T}$-tiling of $\mathbb{N} \times \mathbb{N}$. It is well-known that this problem is $\Pi_1^0$-complete.

We will start by proving that adding cyclic permutation to one-dimensional fluted logic leads to an undecidable logic $\mathrm{GRA}(p,\neg,\Dot{\cap},\exists_1,\exists_0)$. Define the standard grid $\mathfrak{G}_\mathbb{N} := (\mathbb{N} \times \mathbb{N},R,U)$ where we have $R = \{((i,j),(i+1,j)) \mid i,j \in \mathbb{N}\}$ and $U = \{((i,j),(i,j+1)) \mid i,j \in \mathbb{N}\}$.

Consider the extended vocabulary $\{R,U,L,D,S\}$, where $L$ and $D$ are binary relation symbols while $S$ is a quaternary relation. Define
\[\varphi_{successor} := \forall v_1 (\exists v_2 R(v_1,v_2) \land \exists v_2 U(v_1,v_2))\]
\[\varphi_{inverses} := \forall v_1 \forall v_2 (R(v_1,v_2) \leftrightarrow L(v_2,v_1)) \land \forall v_1 \forall v_2 (U(v_1,v_2) \leftrightarrow D(v_2,v_1))\]
\[\varphi_{cycle} := \forall v_1 \forall v_2 \forall v_3 \forall v_4 [(L(v_1,v_2) \land U(v_2,v_3) \land R(v_3,v_4)) \to S(v_1,v_2,v_3,v_4)] \]
\[\varphi_{completion} := \forall v_1 \forall v_2 \forall v_3 \forall v_4 (S(v_1,v_2,v_3,v_4) \to D(v_4,v_1))\]
Define $\Gamma := \varphi_{successor} \land \varphi_{inverses} \land \varphi_{cycle} \land \varphi_{completion}$. The intended model of $\Gamma$ is the standard grid $\mathfrak{G}_\mathbb{N}$ extended with two binary relations, $L$ pointing left and $D$ pointing down, together with a one quaternary relation $S$ which contains all the cycles.

\begin{lemma}
    Let $\mathfrak{G}$ be a structure of the vocabulary $\{R,U,L,D,S\}$. Suppose $\mathfrak{G}$ satisfies $\Gamma$. Then there exists a homomorphism from $\mathfrak{G}_\mathbb{N}$ to $\mathfrak{G} \upharpoonright \{R,U\}$, i.e. to the restriction of $\mathfrak{G}$ to the vocabulary $\{R,U\}$.
\end{lemma}
\begin{proof}
    Since $\mathfrak{G}$ satisfies $\varphi_{inverses}, \varphi_{cycle}$ and $\varphi_{completion}$, it is clear that $\mathfrak{G}$ also satisfies the sentence
    \[\varphi_{grid-like} := \forall v_1 \forall v_2 \forall v_3 \forall v_4 [(R(v_1,v_2) \land U(v_1,v_3) \land R(v_3,v_4)) \to U(v_2,v_4)].\]
    Using this sentence together with $\varphi_{successor}$, it is easy to inductively construct a homomorphism from $\mathfrak{G}_\mathbb{N}$ to $\mathfrak{G} \upharpoonright \{R,U\}$.
\end{proof}

\begin{theorem}
    The satisfiability problem for $\mathrm{GRA}(p,\neg,\Dot{\cap},\exists_1,\exists_0)$ is $\Pi_1^0$-complete.
\end{theorem}
\begin{proof}
    It suffices to show that $\Gamma$ can be expressed in $\mathrm{GRA}(p,\neg,\Dot{\cap},\exists_1,\exists_0)$, since it is routine to construct an encoding of the tiling problem using the sentence $\Gamma$ (see for example \cite{preprintofthis2}). Clearly $\varphi_{successor}$ and $\varphi_{inverses}$ can be expressed in $\mathrm{GRA}(p,\neg,\Dot{\cap},\exists_1,\exists_0)$. The formula $\varphi_{cycle}$ can be expressed as the term
    \[\forall_0 pp(\neg L \Dot{\cup} p(\neg U \Dot{\cup} p(\neg R \Dot{\cup} S))\]
    while $\varphi_{completion}$ can be expressed as the term
    \[\forall_0 p(ppp \neg S \Dot{\cup} D).\]
    Thus $\Gamma$ can be expressed in $\mathrm{GRA}(p,\neg,\Dot{\cap},\exists_1,\exists_0)$.
\end{proof}

We will then prove that adding the operator $s$ to the fluted logic leads to an undecidable logic. Let $\mathbb{T}$ be a set of tiles. Consider the following vocabulary
\[\{P_t \mid t\in \mathbb{T}\} \cup \{T_t^H \mid t\in \mathbb{T}\} \cup \{T_t^V \mid t \in \mathbb{T}\},\]
where each $P_t$ is a binary relational symbol, while $T_S$ and $T_t^H$ are ternary relational symbols. Define
\[\varphi_1 := \forall v_1 \forall v_2 \bigvee_{t\in \mathbb{T}} P_t(v_1,v_2)\]
\[\varphi_2 := \forall v_1 \forall v_2 \bigwedge_{t\neq t'} \neg (P_t (v_1,v_2) \land P_{t'} (v_1,v_2))\]
\[\varphi_3 := \forall v_1 \exists v_2 \forall v_3 (\bigvee_{t_R = t_L'} (T_t^H(v_1,v_2,v_3) \land P_{t'}(v_2,v_3)) \land \bigvee_{t_T = t_D'} (T_t^V(v_1,v_2,v_3) \land P_{t'}(v_3,v_2)))\]
\[\varphi_4 := \forall v_1 \forall v_3 \bigwedge_{t\in \mathbb{T}} (\neg \exists v_2 T_t^H(v_1,v_2,v_3) \lor P_t(v_1,v_3))\]
\[\varphi_5 := \forall v_1 \forall v_3 \bigwedge_{t\in \mathbb{T}} (\neg \exists v_2 T_t^V(v_1,v_2,v_3) \lor P_t(v_3,v_1))\]
Define $\Gamma_{\mathbb{T}} := \varphi_1 \land \varphi_2 \land \varphi_3 \land \varphi_4 \land \varphi_5$. We remark here that this encoding of the tiling problem is essentially the same as the one used in \cite{borger97} to prove that the Kahr-fragment is a conservative reduction class. For completeness we will sketch a proof for the fact that the encoding is correct.

\begin{comment}
    Let $\mathcal{P}$ denote the following term of this vocabulary
    \[\mathcal{P}_1 \Dot{\cap} \mathcal{P}_2 \Dot{\cap} \mathcal{P}_3 \Dot{\cap} \mathcal{P}_4 \Dot{\cap} \mathcal{P}_5,\]
    where the terms are defined as follows:
    \[\mathcal{P}_1 := \forall \forall \Dot{\bigcup_{t}} P_t\]
    \[\mathcal{P}_2 := \forall \forall \Dot{\bigcap_{t\neq t'}} \neg (P_t \Dot{\cap} P_{t'})\]
    \[\mathcal{P}_3 := \forall \exists \forall (\Dot{\bigcup_{t_R = t_L'}}(P_{t'} \Dot{\cap} T_t^H) \Dot{\cap} \Dot{\bigcup_{t_T = t_D'}} (s P_{t'} \Dot{\cap} T_t^V))\]
    \[\mathcal{P}_4 := \Dot{\bigcap_{t}} \forall \forall (\neg \exists s T_t^H \Dot{\cup} P_t)\]
    \[\mathcal{P}_5 := \Dot{\bigcap_{t}} \forall \forall (\neg \exists s T_t^V \Dot{\cup} s P_t)\]
\end{comment}

\begin{lemma}
    $\Gamma_{\mathbb{T}}$ is satisfiable if and only if $\mathbb{T}$ tiles $\mathbb{N} \times \mathbb{N}$.
\end{lemma}
\begin{proof}
    The direction from right to left is clear, so we will focus on the direction from left to right. Suppose that there exists a model $\mathfrak{A}$ so that $\mathfrak{A} \models \Gamma_{\mathbb{T}}$. Let $a\in A$ be an arbitrary element. Since $\mathfrak{A} \models \varphi_3$, we know that there exists a function $S:A\to A$ with the property that for every $a,b\in A$, $(a,S(a),b)$ belongs to the interpretation of the quantifier-free part of $\varphi_3$. Let $S^n(a)$ denote the element obtained from $a$ by applying $S$ to it $n$-times. We can now define a tiling $\sigma : \mathbb{N} \times \mathbb{N} \to \mathbb{T}$ by setting that $\sigma(n,m) = t$ if and only if $(S^n(a),S^m(a)) \in P_t^\mathfrak{A}$. Note that $\varphi_1$ and $\varphi_2$ guarantee that this is a well-defined function. Furthermore it is straightforward to check that the sentences $\varphi_3, \varphi_4$ and $\varphi_5$ guarantee that this is a valid tiling.
\end{proof}

\begin{theorem}
    The satisfiability problem for $\mathrm{GRA}(s,\neg,\Dot{\cap},\exists)$ is $\Pi_1^0$-complete.
\end{theorem}
\begin{proof}
    It suffices to show that $\Gamma_{\mathbb{T}}$ can be expressed in $\mathrm{GRA}(s,\neg,\Dot{\cap},\exists)$. First we note that clearly $\varphi_1$ and $\varphi_2$ can be expressed in $\mathrm{GRA}(s,\neg,\Dot{\cap},\exists)$. For $\varphi_3$ we can use the following term
    \[\forall \exists \forall (\Dot{\bigcup_{t_R = t_L'}}(P_{t'} \Dot{\cap} T_t^H) \Dot{\cap} \Dot{\bigcup_{t_T = t_D'}} (s P_{t'} \Dot{\cap} T_t^V)),\]
    while for $\varphi_4$ and $\varphi_5$ we can use the following terms
    \[\forall \forall \Dot{\bigcap_{t}} (\neg \exists s T_t^H \Dot{\cup} P_t)\]
    \[\forall \forall \Dot{\bigcap_{t}} \forall \forall (\neg \exists s T_t^V \Dot{\cup} s P_t).\]
    Thus $\Gamma$ can be expressed in $\mathrm{GRA}(s,\neg,\Dot{\cap},\exists)$.
\end{proof}

We note that the above proof is optimal in the sense that we used only relational symbols of arity at most three, while it is known that $\mathrm{GRA}(s,\neg,\Dot{\cap},\exists)$ is decidable over vocabularies with at most binary relational symbols, since it can be translated to $\mathrm{FO}^2$ \cite{preprintofthis2}. Another thing to note is that if we were to introduce a \emph{two-dimensional} intersection, which would be defined analogously with the one-dimensional intersection, then adding that operator together with $C$ to $\mathrm{GRA}(s,\neg,\cap,\exists)$ would result in an undecidable logic.

\section{Conclusion}

In this article we have studied in detail how the complexities of various ordered fragments of first-order logic change if we modify slightly the underlying syntax. The general picture that emerges is that even if we relax only slightly the restrictions on the syntax, the complexity of the logic can increase drastically. On the other hand, we have seen that adding the further restriction of one-dimensionality on the logics can greatly decrease the complexity of the logic.

There are several directions in which the work conducted in this article can be continued. Perhaps the most immediate problem is whether the extension of ordered logic $\mathrm{GRA}(s,E,\neg,C,\cap,\exists)$ is decidable. As we have seen, this logic does not have the finite model property, and thus we can't expect that traditional model building techniques can be used to prove that it is decidable. On the other hand, we have not been able to prove that this logic is undecidable using standard tiling arguments.

Another direction for future research is to try to identify possible applications for the ordered logics. For now the main motivation for their study has been based on purely theoretical reasons, but in principle these logics should provide possibilities for new applications of decidable fragments of first-order logic, since their expressive power is orthogonal to the expressive power of other well-studied fragments of first-order logic.

\section*{Acknowledgement}

The research leading to this work was supported by the Academy of Finland grants 324435 and 328987. The author also wishes to thank Antti Kuusisto for many helpful discussions on fragments of first-order logic.

\bibliographystyle{plainurl}
\bibliography{kirjasto}

\end{document}